\documentclass[a4paper]{article}
\usepackage[margin=1.25in]{geometry}

\usepackage{authblk}
\usepackage[utf8]{inputenc}
\usepackage{microtype}
\usepackage{amsfonts}
\usepackage{amssymb}
\usepackage{amsmath}
\usepackage{amsthm}
\usepackage{thm-restate}
\usepackage{comment}
\usepackage{xspace}
\usepackage{enumitem}
\usepackage[hidelinks]{hyperref}
\usepackage[capitalize]{cleveref}
\usepackage[ruled,linesnumbered,noend]{algorithm2e}

\newtheorem{theorem}{Theorem}
\newtheorem{lemma}[theorem]{Lemma}

\newtheorem{fact}[theorem]{Fact}
\newtheorem{observation}[theorem]{Observation}
\newtheorem{question}[theorem]{Question}
\newtheorem{corollary}[theorem]{Corollary}

\crefname{claim}{Claim}{Claims}

\title{The Streaming $k$-Mismatch Problem: Tradeoffs between Space and Total Time\thanks{This work was supported in part by ISF grants no.\ 1278/16 and 1926/19, by a BSF grant no.\ 2018364, and by an ERC grant MPM under the EU's Horizon 2020 Research and Innovation Programme (grant no. 683064).}}

\newcommand{\ceil}[1]{\left\lceil{#1}\right\rceil}
\newcommand{\floor}[1]{\left\lfloor{#1}\right\rfloor}

\newcommand{\modulo}{\operatorname{mod}}
\newcommand{\naive}{na\"{\i}ve }		
\newcommand{\HAM}{\mathsf{Ham}}
\newcommand{\ALG}{\mathsf{ALG}}
\newcommand{\num}{t}

\newcommand{\diam}{\mathsf{diam}}

\newcommand{\Z}{\mathbb{Z}}

\newcommand{\I}{\Phi}

\newcommand{\F}{\mathcal{F}}
\newcommand{\G}{\mathcal{G}}
\newcommand{\X}{\mathcal{X}}
\newcommand{\Y}{\mathcal{Y}}
\renewcommand{\H}{\mathcal{H}}

\newcommand{\Oh}{O}
\newcommand{\Ohtilde}{\tilde{\Oh}}

\newcommand{\supp}{\mathrm{supp}}

\newcommand{\eps}{\varepsilon}

\newcommand{\dd}{\mathinner{.\,.}}

\begin{document}
\allowdisplaybreaks

\setenumerate{itemsep=0ex, parsep=1pt, topsep=1pt}

\renewcommand\Affilfont{\normalsize}

\author[]{Shay Golan}
\author[]{Tomasz Kociumaka}
\author[]{Tsvi Kopelowitz}
\author[]{Ely Porat}

\affil[]{Department of Computer Science, Bar-Ilan University, Ramat Gan, Israel}
\affil[]{\texttt{ golansh1@cs.biu.ac.il, kociumaka@mimuw.edu.pl,  kopelot@gmail.com, porately@cs.biu.ac.il}}
\date{\vspace{-1.5cm}}

\maketitle

\begin{abstract}
We revisit the $k$-mismatch problem in the streaming model on a pattern of length $m$ and a streaming text of length $n$, both over a size-$\sigma$ alphabet.
The current state-of-the-art algorithm for the streaming $k$-mismatch problem, by Clifford et al. [SODA 2019], uses $\Ohtilde(k)$ space and $\Ohtilde\big(\sqrt k\big)$ worst-case time per character. The space complexity is known to be (unconditionally) optimal, and the worst-case time per character matches a conditional lower bound. However, there is a gap between the total time cost of the algorithm, which is $\Ohtilde(n\sqrt k)$, and the fastest known offline algorithm, which costs $\Ohtilde\big(n + \min\big(\frac{nk}{\sqrt m},\sigma n\big)\big)$ time.
Moreover, it is not known whether improvements over the $\Ohtilde(n\sqrt k)$ total time are possible when using more than $O(k)$ space.

We address these gaps by designing a randomized streaming algorithm for the $k$-mismatch problem that,
given an integer parameter $k\le s \le m$, uses $\tilde O(s)$ space and costs $\tilde O\big(n+\min\big(\frac {nk^2}m,\frac{nk}{\sqrt s},\frac{\sigma nm}s\big)\big)$ total time.
For $s=m$, the total runtime becomes $\Ohtilde\big(n + \min\big(\frac{nk}{\sqrt m},\sigma n\big)\big)$, which matches the time cost of the fastest offline algorithm.
Moreover, the worst-case time cost per character is still $\tilde O\big(\sqrt k\big)$.
\end{abstract}

\section{Introduction}

In the fundamental Hamming distance problem, given two same-length strings $X$ and $Y$, the goal is to compute $\HAM(X,Y)$, which is the number of aligned mismatches between $X$ and $Y$.
In the pattern matching version of the Hamming distance problem, the input is a pattern $P$ of length $m$ and a text $T$ of length $n$, both over a size-$\sigma$ alphabet,
and the goal is to compute the Hamming distance
between $P$ and every length-$m$ substring of $T$.
In this paper, we focus on a well studied generalization known as the \emph{$k$-mismatch} problem~\cite{Abrahamson87,ALP04,CGKKP20,CFPSS16,CH02,FP74,GG87,GU18,LV86,LV89,SV96}, which is the ``fixed-threshold'' version of the pattern matching Hamming distance problem: for a given parameter $k$, for each length-$m$ substring $S$ of $T$, if $\HAM(P,S) \le k$, then compute $\HAM(P,S)$, and otherwise, report that $\HAM(P,S) > k$.
Currently, the state-of-the-art (offline) algorithms for the $k$-mismatch problem are: (1) the algorithm of Fischer and Paterson~\cite{FP74}, whose runtime is $\Ohtilde (\sigma n)$, and (2) the algorithm of Gawrychowski and Uznański~\cite{GU18}, whose runtime is $\Ohtilde\big(n+ \frac{nk}{\sqrt{m}}\big)$; see also~\cite{CGKKP20}.

\paragraph{The online and streaming models.}
The growing size of strings to be processed, often exceeding the available memory limits, motivated
the study of pattern matching in the streaming model, where the characters of $T$ arrive in a stream one at a time,
and every occurrence of $P$ needs to be identified as soon as the last character of the occurrence arrives~\cite{BG14,CGKKP20,CFPSS15,CFPSS16, CKP19, GS19,GKP18,GP17, RS17,PP09,SSU19}.
In the \emph{streaming $k$-mismatch problem}, the goal is to compute the Hamming distance between $P$ and the current length-$m$ suffix of $T$ after each new character arrives, unless the Hamming distance is larger than $k$ (which the algorithm reports in this case).
Algorithms in the streaming model are typically required to use space of size sublinear in $m$.
A closely related model is the \emph{online model}, where the space usage of the algorithm is no longer explicitly limited.

Porat and Porat~\cite{PP09} introduced the first streaming $k$-mismatch algorithm
using $\Ohtilde(k^2)$ time per character and $\Ohtilde(k^3)$ space. Subsequent improvements~\cite{CFPSS16,GKP18}
culminated in an algorithm by Clifford et al.~\cite{CKP19}, which solves the streaming $k$-mismatch problem in $\Ohtilde\big(\sqrt{k}\big)$ time per character using $\Ohtilde(k)$ space.
The total time cost of $\Ohtilde\big(n\sqrt k\big)$ matches the time cost of the offline algorithm of Amir et al.~\cite{ALP04}, and the worst-case per-character running time matches a recent lower bound (valid for $\sigma=\Omega\big(\sqrt{k}\big)$ even with unlimited space usage) by Gawrychowski and Uznański~\cite{GU19}, conditioned on the combinatorial Boolean matrix multiplication conjecture.
However, the $\Ohtilde\big(n+ \frac{nk}{\sqrt m}\big)$  total time cost of the offline algorithm of Gawrychowski and Uznański~\cite{GU18} is smaller, and the $\Ohtilde(\sigma n)$  total time cost of the offline algorithm of
Fischer and Paterson~\cite{FP74} is smaller for small $\sigma$.

In the online model, where $O(m)$ space usage is allowed,
the fastest algorithms follow from a generic reduction by Clifford et al.~\cite{CEPP11},
which shows that if the offline $k$-mismatch problem can be solved in $\Oh(n\cdot t(m,k))$ time,
then the online $k$-mismatch problem can be solved in $\Oh(n\sum_{i=0}^{\ceil{\log m}}t(2^i,k))$ time.
In particular, this yields online algorithms with a total runtime of $\Ohtilde\big(n\sqrt{k}\big)$ and $\Ohtilde(n\sigma)$.
Nevertheless, this approach cannot benefit from the state-of-the-art the offline algorithm of Gawrychowski and Uznański~\cite{GU18} since the running time of this algorithm degrades as $m$ decreases.
Thus, a natural question arises:
\begin{question}\label{que:total-time}
Is there an online/streaming algorithm for the $k$-mismatch problem whose total time cost is $\Ohtilde\big(n+\min\big(\frac{nk}{\sqrt m},\sigma n\big)\big)$?
\end{question}

\paragraph{Space usage.}
It is straightforward to show that any streaming algorithm for the $k$-mismatch problem must use $\Omega(k)$ space~\cite{CKP19}.
Thus, the space usage of the algorithm of Clifford et al.~\cite{CKP19} is optimal.
Remarkably, 
we are unaware of any other tradeoffs between (sublinear) space usage and runtime for the $k$-mismatch problem. This leads to the following natural question.

\begin{question}\label{que:space-tradeoff}
Is there a time-space tradeoff algorithm for the $k$-mismatch problem, using $s \ge \Omega(k)$ space?
\end{question}

\paragraph{Our results.}
We address both \cref{que:total-time} and \cref{que:space-tradeoff} by proving the following theorem.

\begin{theorem}\label{thm:kmm-total}
	There exists a randomized streaming algorithm for the $k$-mismatch problem that, given an integer parameter $k\le s \le m$, costs $\tilde O\big(n+\min\big(\frac {nk^2}m,\frac{nk}{\sqrt s},\frac{\sigma nm}s\big)\big)$ total time and uses $\tilde O(s)$ space. Moreover, the worst-case time cost per character is $\tilde O(\sqrt k)$. The algorithm is correct with high probability\footnote{An event $\mathcal E$ happens with high probability if  $\Pr[\mathcal E]\ge 1-n^{-c}$ for a constant parameter $c\ge 1$.}.
\end{theorem}

\cref{thm:kmm-total} answers \cref{que:space-tradeoff} directly. However, for \cref{que:total-time}, \cref{thm:kmm-total} only addresses the online setting, where $s=m$ can be set: since $k<\sqrt m$ yields $n > \frac{nk}{\sqrt m} > \frac {nk^2}m$, the total time cost is  $\tilde O\big(n+\min\big(\frac {nk^2}m,\frac{nk}{\sqrt m},\frac{\sigma nm}m\big)\big) = \tilde O\big(n+\min\big(\frac{nk}{\sqrt m},\sigma n\big)\big)$.
However, \cref{que:total-time} remains open for the streaming model.

Another natural research direction is to extend \cref{thm:kmm-total} so that the pattern $P$ 
could also be processed in a streaming fashion using $\Ohtilde(s)$ space, $\Ohtilde(\sqrt{k})$ time per character,
and $\tilde O\big(m+\min\big(k^2,\frac{mk}{\sqrt s},\frac{\sigma m^2}s\big)\big)$  time in total.
To the best of our knowledge, among the existing streaming $k$-mismatch algorithms, 
only that of Clifford et al.~\cite{CKP19} is accompanied with an efficient streaming procedure
for preprocessing the pattern.

\section{Algorithmic Overview and Organization}\label{sec:preliminaries}

A string $S$ of length $|S|=n$ is a sequence of characters $S[0]S[1]\cdots S[n-1]$ over an alphabet $\Sigma$.
A \emph{substring} of $S$ is denoted by $S[i\dd j]=S[i]S[i+1]\cdots S[j]$ for $0\leq i \leq j < n$. If $i=0$, the substring is called a \emph{prefix} of $S$, and if $j=n-1$, the substring is called a \emph{suffix} of $S$.
For two strings $S$ and $S'$ of the same length $|S|=n=|S'|$, we denote by $\HAM(S,S')$ the Hamming distance of $S$ and~$S'$,
that is, $\HAM(S,S') = |\{0\le i \le n-1 : S[i] \ne S'[i]\}|$.
An integer $\rho$ is a $d$-period of a string $S$ if $\HAM(S[0\dd  n-\rho-1], S[\rho \dd  n-1])\le d$.

\subsection{Overview}\label{sec:overview}
To prove \cref{thm:kmm-total}, we consider two cases, depending on whether or not there exists an integer $\rho\le k$ that is a $d$-period of $P$ for some $d=\Oh(k)$.
If such a $\rho$ exists, then we say that $P$ is \emph{periodic}\footnote{The classic notion of periodicity is usually much simpler than the one we define here. However, since in this paper we do not use the classic notion of periodicity, we slightly abuse the terminology.}, and otherwise $P$ is said to be \emph{aperiodic}.

\paragraph{Tail partitioning.}
In both cases of whether $P$ is periodic or not, our algorithms use the well-known \emph{tail partitioning} technique~\cite{CFPSS15,CFPSS16,CKP19,CS10,GP17}, which decomposes $P$  into two substrings: a suffix $P_{tail}$ and the complementary prefix $P_{head}$ of length $m-|P_{tail}|$.
Accordingly, the algorithm has two components.
The first component computes the Hamming distance of $P_{head}$ and every length-$|P_{head}|$ substring of $T$ with some delay: the reporting of $\HAM(P_{head},\allowbreak T[i-|P|+1\dd i-|P_{tail}|])$ is required to be completed before the arrival of $T[i]$.
The second component computes the Hamming distance of $P_{tail}$ and carefully selected length-$|P_{tail}|$ substrings of $T$.
The decision mechanism for selecting substrings for the second component is required to guarantee that whenever $\HAM(P_{head},T[i-|P|+1\dd i-|P_{tail}|])\le k$:
if $\HAM(P_{tail},T[i-|P_{tail}|+1\dd i])\le k$ then the second component computes $\HAM(P_{tail},\allowbreak T[i-|P_{tail}|+1\dd i])$; otherwise, the second component reports $\HAM(P_{tail},T[i-|P_{tail}|+1\dd i])>k$.
The second component has no delay.

Notice that if either $\HAM(P_{head},T[i-|P|\dd i-|P_{tail}|])>k$, which is detected by the first component,
or $\HAM(P_{tail},T[i-|P_{tail}|+1\dd i]) >k$, which is detected by the second component, then it must be that
$\HAM(P,T[i-|P|\dd i])> k$.
Otherwise, $\HAM(P,T[i-|P|\dd i])$ is computed by summing $\HAM(P_{head},T[i-|P|\dd i-|P_{tail}|])$ 
and  $\HAM(P_{tail},T[i-|P_{tail}|+1\dd i])$. 
In either case, the information is available for the algorithm right after $T[i]$ arrives.

Thus, our algorithm has four main components, depending on whether $P$ is periodic or not, and depending on the head or tail case of the tail partitioning technique.

\paragraph{The aperiodic case.}
The algorithms for the aperiodic case are a combination of straightforward modifications of previous work together with the \naive algorithm; the details are given in \cref{sec:aperiodic-pattern}. Nevertheless, we provide an overview below. In this case, $|P_{tail}|=2k$.

The algorithm for $P_{head}$ in the aperiodic case is a slight modification of an algorithm designed by Golan et al.~\cite{GKP18}, which reduces the streaming $k$-mismatch problem to the problem of finding occurrences of multiple patterns in multiple text-streams.

The algorithm for $P_{tail}$ in the aperiodic case is the \naive algorithm of comparing all aligned pairs of characters. While in general the \naive algorithm could cost $O(|P_{tail}|)$ time per character, in our setting the algorithm uses the output of the algorithm on $P_{head}$ as a filter, and so the algorithm computes $\HAM(P_{tail}, T[i-|P_{tail}|+1\dd i])$ only if $\HAM(P_{head},\allowbreak T[i-|P|+1\dd i-|P_{tail}|]) \le k$.
Since $P$ is aperiodic, we are able to show that occurrences of $P_{head}$ are distant enough so that the \naive algorithm for $P_{tail}$ costs $\Ohtilde(1)$ worst-case time per character.
In order for the filter to be effective, instead of guaranteeing that the algorithm for computing $\HAM(P_{head},T[i-|P|+1\dd i-|P_{tail}|])$ is completed before $T[i]$ arrives, we refine the tail partitioning technique so that the computation of $\HAM(P_{head},T[i-|P|+1\dd i-|P_{tail}|])$ completes before $T[i - \frac12|P_{tail}|]$ arrives, and if the \naive algorithm should be used, the execution takes place through the arrivals of the subsequent $\frac12|P_{tail}|$ characters $T[i - \frac12|P_{tail}|+1\dd i]$. The effects of this refinement on the runtime is only by constant multiplicative factors.

\paragraph{The periodic case.}
We begin by first assuming that $P$ and $T$ have a common $O(k)$-period $\rho \le k$, and that $n\le\frac32m$.
In this case, we represent the strings as characteristic functions (one function for each character in $\Sigma$).
Since both $P$ and $T$ are assumed to be periodic, each characteristic function, when treated as a string, is also periodic.
Next, we use the notion of \emph{backward differences}: for any function $f:\Z \rightarrow \Z$, the \emph{backward difference} of $f$ due to $\rho$ is $\Delta_\rho[f](i)= f(i)-f(i-\rho)$.
Clifford et al.~\cite{CKP19} showed that the Hamming distance of two strings can be derived from a summation of convolutions of backward differences due to $\rho$ of characteristic functions; see \cref{sec:hd-to-csp}.

In the case of $P_{head}$, a delay of up to $2s$ characters is allowed.
To solve this case, we define the problem of computing the convolutions of the backward differences in batches; the details for this case are given in~\cref{sec:periodic-everything-with-delay}.
Our solution uses an offline algorithm for  computing the convolutions described in \cref{sec:csp}. 
In the case of $P_{tail}$, we use a solution for the online version of computing the convolutions of the backward differences, which is adapted from Clifford et al.~\cite{CEPP11}; the details are given in \cref{sec:periodic-everything-without-delay}.
In both cases, since we assume that $P$ and $T$ are periodic, our algorithms leverage the fact that the backward differences of the characteristic functions have a small number of non-zero entries.
This lets the algorithms compute the Hamming distance of $P_{tail}$ and every substring of $T$ which has length $|P_{tail}|$.


In \cref{sec:periodic-pattern-aribtrary-text}, we remove the periodicity assumption on $T$ by applying a technique by Clifford et al.~\cite{CFPSS16} which identifies at most one periodic region of $T$ that contains all the $k$-mismatch occurrences of $P$. Moreover, we drop the $n\le \frac32m$ assumption using a standard trick of partitioning $T$ into overlapping fragments of length $\frac32m$.

\section{Hamming Distance and the Convolution Summation \\ Problem}\label{sec:hd-to-csp}
Recall that the \emph{support} of a function $f$ is $\supp(f):=\{x \mid f(x)\ne 0\}$.
Let $|f|=|\supp(f)|$.
Throughout, we only consider functions with finite support mapping $\Z$ to $\Z$.
The \emph{convolution} of two functions $f,g:\Z\rightarrow \Z$  is a function $f\ast g:\Z\rightarrow \Z$ such that \[[f\ast g](i)=\sum_{j\in\Z} f(j)\cdot g(i-j).\]

For a string $X$ and a character $c\in \Sigma$,  the \emph{characteristic function} of $X$ and $c$ is $X_c:\Z\rightarrow\{0,1\}$ such that $X_c(i)=1$ if and only if $X[i]=c$.
For a string $X$, let $X^R$ be $X$ reversed.
The \emph{cross-correlation} of strings $X$ and $Y$ over $\Sigma$ is a function $X\otimes Y:\Z\rightarrow\Z$ such that
\[X\otimes Y=\sum_{c\in\Sigma} X_c\ast Y^R_c.\]

\begin{lemma}[{\cite[Fact 7.1]{CKP19}}]\label{lem:crosscorrelationHam}
	Let $P,T$ be strings.
	For $|P|-1\le i<|T|$, we have $[T\otimes P](i)=|P|-\HAM(P,T[i-|P|+1\dd i])$.
	For $i<0$ and for $i\ge|P|+|T|$, we have $[T\otimes P](i)=0$.
\end{lemma}
By \cref{lem:crosscorrelationHam}, in order to compute $\HAM(P,T[i-|P|+1\dd i])$, it suffices to compute $[T\otimes P](i)$.

The \emph{backward difference} of a function $f:\Z \to \Z$ due to $\rho$ is $\Delta_\rho[f](i)= f(i)-f(i-\rho)$.
\begin{observation}[{\cite[{Obs.~7.2}]{CKP19}}]\label{obs:period-length-small-norm}
	If a string $X$ has a $d$-period $\rho$, then $\sum_{c\in \Sigma}|\Delta_\rho[X_{c}]| \le 2(d+\rho)$.
\end{observation}

Our computation of $T\otimes P$ in a streaming fashion is based on the following lemma:
\begin{lemma}[Based on {\cite[Fact 7.4 and Corollary 7.5]{CKP19}}]\label{lem:crosscorrelation_computation}
	For every $i\in \mathbb{Z}$ and $\rho\in \mathbb{Z}_+$, we have
	$[T\otimes P](i)=\big[\sum_{c\in\Sigma}\Delta_\rho[T_c]\ast\Delta_\rho[P^R_c]\big](i)-[T\otimes P](i-2\rho)+2[T\otimes P](i-\rho)$.
\end{lemma}

When computing $[T\otimes P](i)$, if the algorithm maintains a buffer of the last $2\rho$ values of $T\otimes P$, then the algorithm already has the values of $[T\otimes P](i-\rho)$ and $[T\otimes P](i-2\rho)$.
Thus, in order to construct $T\otimes P$, the focus is on constructing $\sum_{c\in\Sigma}\Delta_\rho[T_c]\ast\Delta_\rho[P^R_c]$.

\subsection{Convolution Summation Problem}\label{sec:csp}

We express the task of constructing $\sum_{c\in\Sigma}\Delta_\rho[T_c]\ast\Delta_\rho[P^R_c]$ in terms of a more abstract \emph{convolution summation} problem stated as follows.
The input is two sequences of functions $\mathcal F=(f_1,f_2,\ldots,f_\num)$ and $\mathcal G=(g_1,g_2,\ldots g_\num)$ such that for every $1\le i\le \num$ we have $f_i:\Z\rightarrow\Z$ and $g_i:\Z\rightarrow\Z$,
and the goal is to construct the function $\mathcal F \otimes \mathcal G = \sum_{j=1}^\num(f_j\ast g_j)$.

Let $\H$ be a sequence of functions.
We define the \emph{support} of  $\mathcal H$ as $\supp(\mathcal H)=\bigcup_{h\in\mathcal H}\supp(h)$.
The total number of non-zero entries in all of the functions of  $\mathcal H$ is denoted by
$\|\mathcal H\|= \sum_{h\in\mathcal H}|h|$.
The \emph{diameter} of a  function $f$ is $\diam(f)=\max(\supp(f))-\min(\supp(f))+1$ if $\supp(f)\ne\emptyset$, and $\diam(f)=0$ otherwise.
We define the \emph{diameter} of a sequence of functions $\mathcal H $ as
$\diam(\mathcal H)=\max\left\{\diam(h)\mid h\in\H\right\}$.

In our setting, the input for the convolution summation problem is two sequences of \emph{sparse functions}, which are functions that have a small support.
Thus, we assume that the input functions are given in an efficient \emph{sparse representation}, e.g., a linked list (of the non-zero functions) of linked lists (of non-zero entries).

\paragraph*{Algorithm for the offline convolution summation problem.}
The following lemma, provides an algorithm that efficiently computes $\mathcal{F\otimes G}$ for two sequences of functions $\F$ and $\G$ which are given in a sparse representation.
Notice that the output of the algorithm is also restricted to the non-zero values of $\mathcal{F\otimes G}$ only.

\begin{restatable}[{Based on~\cite[{Lemma 7.5}]{CGKKP20}}]{lemma}{offlineconvolution}\label{lem:offline-convolution}
	Let	$\mathcal{F}=(f_1,\ldots,f_\num)$ and $\mathcal{G}=(g_1,\ldots,g_\num)$ be two sequences of functions, such that
	$\diam(\F),\diam(\G)\in [1\dd n]$, and also $\diam(\F\otimes \G)=O(n)$.
	Then there exists an (offline) algorithm that computes the non-zero entries of $\mathcal{F\otimes G}$ using $O(n)$ space, whose time cost is
	\[	\Psi(\F,\G)=\Ohtilde\Bigg(\|\F\|+\|\G\|+\!\!\sum_{j=1}^t  \min(|f_j|  |g_j|,n)\Bigg)=\Ohtilde\left(\min\left(t n, \|\F\|\cdot \|\G\|, (\|\F\|+\|\G\|)\sqrt{n}\right)\right).
\]\end{restatable}
\begin{proof}
	
	There are two methods that the algorithm chooses from to compute each convolution $f_j \ast g_j$.
	The first method is to enumerate all pairs consisting of a non-zero entry in $f_j$ and in $g_j$.
	Using standard dictionary lookup techniques, the time cost of computing the convolution $f_j \ast g_j$ this way is $\Ohtilde(|f_j| |g_j|)$.
	The second method of computing $f_j \ast g_j$ is by FFT, which costs $\Ohtilde(n)$ time.
	The algorithm combines both methods by comparing $|f_j| |g_j|$ to $n$ for each $1\le j \le t$ and picking the cheaper method for each particular $j$. Thus, the time for computing $f_j \ast g_j$ for any $j$ is $\Ohtilde(\min (|f_j| |g_j|,\,n))$.
	Since $\F$ and $\G$ are given in a sparse representation, for $j$ values where both $f_j$ and $g_j$ are zero-functions, the algorithm spends no time at all, while for other $j$s the algorithm spends at least $O(1)$ time, on comparing $|f_j| |g_j|$ to $n$ even if $|f_j| |g_j|=0$.
	
	In order to reduce the space usage, the algorithm constructs $\mathcal{F\otimes G}$
	by iteratively computing
	the sum $\sum_{j=1}^i (f_j\ast g_j)$.
	In each iteration, the algorithm adds the function  $f_j \ast g_j$ to the previously stored sum of functions.
	The summation is stored in a dictionary mapping indices to non-zero values (providing $\Ohtilde(1)$-time access and updates).
	The cost of adding $f_j \ast g_j$ to the previous sum of functions is nearly linear in $|f_j\ast g_j|$ and thus dominated by the time cost of computing $f_j \ast g_j$.
	Hence, the total running time of the algorithm is
	\begin{align*}
	&\Ohtilde\left(\sum_{j\in[1\dd t]:|f_j|+|g_j|>0} (1+\min(|f_j| |g_j|,\,n))\right)\\
	&=\Ohtilde\left(\sum_{j\in[1\dd t]:|f_j|+|g_j|>0} 1+\sum_{j\in[1\dd t]:|f_j|+|g_j|>0}\min(|f_j| |g_j|,\,n)\right)\\
	&=\Ohtilde\left(\|\F\|+\|\G\|+\sum_{j=1}^t\min(|f_j| |g_j|,\,n)\right).
	\end{align*}

	The first bound is obtained by noticing that $\|F\|+\|\G\|=\sum_{j=1}^t|f_j|+|g_j|\le\sum_{j=1}^t 2n=O(tn)$. Thus, \[\Ohtilde\left(\|\F\|+\|\G\|+\sum_{j=1}^t\min(|f_j| |g_j|,\,n)\right)=\Ohtilde\left(tn+\sum_{j=1}^t n\right)=\Ohtilde(t n).\]
	
	Since $\|\F\|\ge1$ and $\|\G\|\ge1$, we have $\|\F\|+\|\G\|=O(\|\F\| \|\G\|)$.
	For each $j$, we have $|f_j|\le \|\F\|$, and therefore
	\[
	\sum_{j=1}^t \min(|f_j|  |g_j|,n) \le \sum_{j=1}^\num |f_j| |g_j|
	\:\le\:\sum_{j=1}^\num \|\F\|  |g_j|
	\:=\:\|\F\|\sum_{j=1}^\num   |g_j|
	\:=\:\|\F\| \|\G\|.
	\]
	The third bound is obtained by recalling that $\min(x,y)\le \sqrt{xy} \le x+y$ holds for every positive $x$ and $y$:
	\[\sum_{j=1}^t \min(|f_j| |g_j|,n)
	\le \sum_{j=1}^t \sqrt{|f_j|  |g_j|n}
	\le \sum_{j=1}^t (|f_j|+|g_j|)\sqrt{n}=(\|\F\|+\|\G\|)\sqrt{n}.\qedhere
	\]
\end{proof}

\section{Periodic Pattern and Text -- with Delay}\label{sec:periodic-everything-with-delay}

Our approach is based on the reduction to the convolution summation problem of \cref{sec:hd-to-csp}. The text arrives online, so we consider a similar setting for convolution summation.

\subsection{The Incremental Batched Convolution Summation Problem}\label{sec:icbsp}
In the incremental batched version of the convolution summation problem, the algorithm is given two sequences of $t$ functions $\F$ and $\G$, where both
$\supp(\F), \supp(\G)\subseteq[0\dd n-1]$.
The sparse representation of $\G$ is available for preprocessing, whereas $\F$ is revealed online in batches of diameter $s$: the $i$th batch consists of all of the non-zero entries of the functions of $\F$ in the range $[(i-1)\cdot s\dd i\cdot s)$, also in a sparse representation.
After each update, the goal is to compute the values of $\F\otimes \G$ in the same range as the input, $[(i-1)\cdot s\dd i\cdot s)$. In the rest of this section, we prove the following lemma.

\begin{lemma}\label{lem:online-batch-algorithm}
There exists a deterministic algorithm that solves the incremental batched convolution summation problem for $s=\Omega(\|\F\|+\|\G\|)$, using $O(s)$ space, $\Ohtilde( (\|\F\|+\|\G\|)\sqrt{s})$ time per batch arrival and $\tilde O\left(n+\min \left(\|\F\|\cdot\|\G\|,\frac {n(\|\F\|+\|\G\|)}{\sqrt s},\frac {tn^2}s\right)\right)$ total time.
\end{lemma}

A natural approach for proving \cref{lem:online-batch-algorithm} is to utilize the algorithm of  \cref{lem:offline-convolution}, whose runtime depends on the diameters of a pair of sequences of functions.
Thus, in order to use this approach, we design a mechanism for reducing the diameters of $\F$ and $\G$ while still being able to properly compute the values of  $\F\otimes\G$.

\paragraph*{Reducing the diameter.}
For a function $h:\Z\rightarrow \Z$ and a domain $D\subseteq \Z$, let $h|_D:\Z \to\Z$ be a function where  $h|_D(i)=h(i)$ for $i\in D$ and $h|_D(i)=0$ for $i\notin D$.
For a sequence of functions $\mathcal H=(h_1,h_2,\ldots, h_\num)$, denote the sequence of functions restricted to domain $D$ as  $\mathcal H|_D=(h_1|_D,h_2|_D,\ldots, h_\num|_D)$.
For $a,b\in \Z$, let us define integer intervals $\Gamma_a=[s\cdot (a-1)\dd s\cdot (a+1))$
and
$\Phi_{b}=[s\cdot (b-1)\dd s\cdot b)$\footnote{Note that the $\Gamma$ intervals are used for partitioning $\G$ while $\Phi$ intervals are used for partitioning $\F$.}.
\begin{observation}\label{obs:partition}
	Each integer is in exactly one range $\Phi_b$ and in exactly two ranges $\Gamma_a$.
\end{observation}

Notice that the $i$th batch consists of $\F|_{\Phi_i}$. In response to the $i$th batch, the algorithm needs to compute $[\F\otimes\G]|_{\Phi_i}$.
For this, we express $\F_i=\F|_{[0\dd si)}=\F|_{\Phi_1\cup\Phi_2\cup\cdots\cup\Phi_i}$ (which is the aggregate of all the batches received so far) and $\G$ in terms of two sequences of functions $\F_i^*$ and $\G^*$, 
so that $[\F\otimes\G]|_{\Phi_i} = [\F_i^*\otimes\G^*]|_{\Phi_i}$.
The motivation for using $\F_i^*$ and $\G^*$ is to reduce the diameter, which comes at the price of increasing the number of functions.

Let $q=\ceil{\frac {n+1}s}$, and define
\begin{align*}
\F^*_i&=\big(f_1|_{\Phi_{i-0}},\ldots,f_t|_{\Phi_{i-0}}
,f_1|_{\Phi_{i-1}},\ldots, f_{\num}|_{\Phi_{i-1}}, \ldots\ldots\ldots, f_1|_{\Phi_{i-(q-1)}},\ldots, f_{\num}|_{\Phi_{i-(q-1)}}\big),\\
\G^*&=\big(g_1|_{\Gamma_{0\phantom{-0}}},\ldots,
g_t|_{\Gamma_{0\phantom{-0}}},g_1|_{\Gamma_{1\phantom{-0}}},\ldots, g_{\num}|_{\Gamma_{1\phantom{-0}}},\ldots\ldots\ldots, g_1|_{\Gamma_{q-1\phantom{-(0)}}},\ldots, g_{\num}|_{\Gamma_{q-1\phantom{-(0)}}}\big).
\end{align*}

\begin{lemma}\label{clm:Fstari-Gstar}
	For any $i\in[1\dd q]$ we have
	$[\F\otimes \G]|_{\Phi_i}=[\F_i^*\otimes \G^*]|_{\Phi_i}$.
\end{lemma}
\begin{proof}
	For two sets $X,Y\subseteq\Z$, we denote $X-Y=\{x-y\mid x\in X\text{ and }y\in Y\}$.
	For any pair of functions $f,g:\Z\rightarrow \Z$ and any $j\in\I_i$,
	since the ranges $\Phi_b$ for $b\in\Z$ form a partition of $\Z$, we have that $(f\ast g )(j)=\sum_{a\in \Z}(f|_{\Phi_{i-a}}\ast g)(j)$. Moreover, by the definition of the convolution operator and since $\Phi_i-\Phi_{i-a}\subseteq \Gamma_a$, we have
	$\sum_{a\in \Z}(f|_{\Phi_{i-a}}\ast g)(j)=\sum_{a\in \Z}(f|_{\Phi_{i-a}}\ast g|_{\Gamma_a})(j)$.
	Hence, 
	$[\F\otimes \G]|_{\Phi_i}=\sum_{a\in \Z}[\F|_{\Phi_{i-a}}\otimes\G|_{\Gamma_a}]|_{\Phi_i}$.
	However, since $\supp(\F)\subseteq [0\dd qs)$, we have that for every $b\le 0$ and every $f\in\F$, it must be that $f|_{\Phi_b}=0$. Similarly, since $\supp(\G)\subseteq [0\dd qs)$, we have that for every	
	$a< 0$  and every $g\in\G$,  it must be that $g|_{\Gamma_a}=0$.
	Thus, for $a\notin [0\dd q-1]$, we have that $\F|_{\Phi_{i-a}}\otimes \G|_{\Gamma_a}=0$, 
	and so $[\mathcal F\otimes \mathcal G]|_{\Phi_i}=\sum_{a=0}^{q-1}[\F|_{\Phi_{i-a}}\otimes\G|_{\Gamma_a}]|_{\Phi_i}$.
	Finally, since $\F_i^*$ is the concatenation of $\F|_{\Phi_{i-a}}$ for $a\in[0\dd q-1]$, 
	and $\G^*$  is the concatenation of  $\G|_{\Gamma_{a}}$ for $a\in[0\dd q-1]$, the lemma follows.
\end{proof}

The following observation is a consequence of the definitions of $\F^*$ and $\G^*$, and \cref{obs:partition}.
Recall that the diameter of a sequence of functions $\mathcal H $ is
$\diam(\mathcal H)=\max\left\{\diam(h)\mid h\in\H\right\}$.

\begin{observation}\label{clm:FG-bounds}
	The sequences $\F^*_i$ and $\G^*$ consists of $\num\cdot q$ functions each.
	Moreover, $\diam(\F_i^*)\le s$, $\diam(\G^*)\le 2s$,  $\|\F^*_i\|\le\|\F\|$, and  $\|\G^*\|\le 2\cdot\|\G\|$.
\end{observation}

\paragraph*{The algorithm.}
During the preprocessing phase, the algorithm transforms $\G$ into $\G^*$, and constructs $\F_0^*$ which is an empty linked list.

Upon receiving the $i$th batch, which is $\F|_{\phi_i}$, the algorithm computes $\F_i^*$ as follows:
First, the algorithm concatenates $\F|_{\Phi_i}$ with $\F_{i-1}^*$ and truncates the last $t$ functions from the resulting sequence of $(q+1)t$ functions.
It is easy to implement the concatenation (including updating indices in the sparse representation) with a time cost which is linear in the number of non-zero functions in $\F|_{\Phi_i}$ and $\F_{i-1}^*$,  and is at most $O(\|\F\|)$.
The implementation of the truncation is void since only the first $(i-1)\cdot t\le (q-1)\cdot t$ functions of $\F_{i-1}^*$ are non-zero\footnote{The reason for mentioning  the truncation, even though it is void, is in order to guarantee that $\F_i^*$ matches the mathematical definition that is used.}.

Next, the algorithm applies the procedure of \cref{lem:offline-convolution} in order to compute $\F_i^*\otimes\G^*$.
Finally, the algorithm returns $[\F_i^*\otimes\G^*]|_{\Phi_i}$ which,  by \cref{clm:Fstari-Gstar}, is $[\F\otimes\G]|_{\Phi_i}$.

\paragraph*{Complexities.}
The preprocessing phase is done in $O(\|\G\|)$ time using $O(s)$ space.

For the $i$th batch, the computation of $\F_i^*$ costs $O(\|\F\|)$ time, which is $O(q\|\F\|)=O(\frac{n}{s}\|\F\|)=O(n)$ time in total because $s=\Omega(\|\F\|+\|\G\|)$.
Then, the computation of $\F_i^*\otimes\G^*$ with the procedure of \cref{lem:offline-convolution} costs $O(\Psi(\F_i^*,\G^*))=\Ohtilde((\|\F_i^*\|+\|\G^*\|)\sqrt{s})=\Ohtilde((\|\F\|+\|\G\|)\sqrt{s})$ time.
In the following we derive several upper bounds on $\sum_{i=1}^{q}\Psi(\F_i^*,\G^*)$, which together upper bound the total time cost.

\paragraph{Total time
$\tilde O(\frac {n(\|\F\|+\|\G\|)}{\sqrt s})$.}
Recall that $q=O(\tfrac ns)$, and, by \cref{clm:FG-bounds}, for any $i\in[1\dd q]$ we have $\|\F_i^*\|\le\|\F\|$ and $\|\G^*\|\le 2 \|\G\|$. Thus,
\begin{multline*}
\sum_{i=1}^{q}\Psi(\F_i^*,\G^*)=
\sum_{i=1}^{q} \Ohtilde\left((\|\F_i^*\|+\|\G^*\|)\sqrt{s}\right)=\Ohtilde\left(\sum_{i=1}^{q} (\|\F\|+\|\G\|)\sqrt{s}\right)=\\
\Ohtilde\left(q\cdot(\|\F\|+\|\G\|)\sqrt{s}\right)=\Ohtilde(
\tfrac ns\cdot(\|\F\|+\|\G\|)\sqrt{s})=\tilde O\left(\tfrac{n(\|\F\|+\|\G\|)}{\sqrt s}\right).
\end{multline*}

\paragraph{Total time
$\tilde O(\frac {tn^2}s)$.}
Recall that for any $i\in[1\dd q]$ the sequences $\G^*$ and $\F_i^*$ consist of $q\cdot t$ functions of diameter $O(s)$, and $q=O(\tfrac ns)$. Thus,
\begin{align*}
\sum_{i=1}^{q}\Psi(\F_i^*,\G^*)= \sum_{i=1}^{q}\Ohtilde\left(q\cdot t\cdot s\right)=\Ohtilde\left(q^2\cdot t\cdot s\right)=\Ohtilde(\tfrac {n^2}{s^2}\cdot t\cdot s)=\Ohtilde\left(\tfrac{tn^2}s\right)
\end{align*}

\paragraph{Total time
	$\tilde O(n+\|\F\|\cdot\|\G\|)$.}
Let $\X = \{x_1,x_2,\ldots,x_\tau\}, \Y = \{y_1,y_2,\ldots,y_\tau\}$ be two sequences of $\tau$ functions from $\Z$ to $\Z$. Let $\nu = \max(\diam (\X),\diam(\Y))$.
Recall that
\[\Psi(\X,\Y)=\Ohtilde\left(\|\X\|+\|\Y\|+\sum_{j=1}^\tau  \min(|x_j|  |y_j|,\,\nu)\right)
=\Ohtilde\left(\|\X\|+\|\Y\|+\sum_{j=1}^\tau  |x_j|  |y_j|\right).\]
In our case, the run time is $\sum_{i=1}^{q}\Psi(\F_i^*,\G^*)$.
Recall that the functions in $\F_i^*$ are
$f_j|_{\Phi_{i-a}}$ for integers $0\le a\le q-1$ and $1\le j \le t$.
Similarly, the functions in $\G^*$ are defined by $g_j|_{\Gamma_a}$ for integer $0\le a\le q-1$.
%
%
Thus,
\begin{align*}
\sum_{i=1}^{q}\Psi(\F_i^*,\G^*)
&=\sum_{i=1}^{q}\tilde O\left(\|\F_i^*\|+\|\G^*\|+\sum_{a=0}^{q-1}\sum_{j=1}^t \Big|f_j|_{\Phi_{i-a}}\Big|  \Big|g_j|_{\Gamma_a}\Big|\right)\\
&=\sum_{i=1}^{q}\tilde O\left(\|\F\|+\|\G\|\right)+\tilde O\left(\sum_{j=1}^t \sum_{a=0}^{q-1} \Big|g_j|_{\Gamma_a}\Big| \sum_{i=1}^{q} \Big|f_j|_{\Phi_{i-a}}\Big|  \right)&\triangleright\text{ by Obs.~\ref{clm:FG-bounds}}\\
%
%
&= \tilde O(\tfrac ns s)+ \tilde O\left(\sum_{j=1}^t \sum_{a=0}^{q-1}  \Big|g_j|_{\Gamma_a}\Big| \Big|f_j\Big|\right)&\triangleright\text{ by Obs.~\ref{obs:partition}}\\
%
%
&= \tilde O(n)+\tilde O\left(\sum_{j=1}^t \Big|g_j\Big|  \Big|f_j\Big|\right)&\triangleright\text{ by Obs.~\ref{obs:partition}}\\
&= \tilde O(n)+\tilde O\left(\|\G\|\sum_{j=1}^t  \Big|f_j\Big|\right) = \tilde O(n)+\tilde O\left(\|\G\|\cdot \|\F\|\right)&\hspace{-22pt}=\tilde O(n+\|\F\|\cdot\|\G\|)
\end{align*}
This completes the proof of \cref{lem:online-batch-algorithm}.\qed%

\subsection{Reduction from Hamming Distance to Incremental Batch Sparse Convolution Summation problem}

Now we show how to compute the Hamming distance between $P$ and substrings of $T$ with delay of $2s$, based on the algorithm of \cref{lem:online-batch-algorithm}.
Our result is stated in the following lemma.

\begin{lemma}\label{lem:kmm-periodic-everything-head}
	Suppose that there exists $\rho\le k$ which is a $d$-period of both $P$ and $T$ for some $d=O(k)$.
	Then, there exists a deterministic streaming algorithm for the $k$-mismatch problem with $n= \frac32m$ that, given an integer parameter $k\le s \le m$, uses $\tilde O(s)$ space and costs $\tilde O\left(m+\min\left(k^2,\frac {mk}{\sqrt s},\frac {\sigma m^2}{s}\right)\right)$ total time. Moreover, the worst-case time cost per character is $\tilde O(\sqrt k)$.
	The algorithm has delay of $2s$ characters.	
\end{lemma}
\begin{proof}
The algorithm receives the characters of $T$ one by one, and splits them into blocks of length $s$ so that the indices in the $b$th block form $\Phi_b$, as defined in \cref{sec:icbsp}.
The last $\Oh(s)$ characters of $T$ and the last $\Oh(s)$ values of $[T\otimes P]$ are buffered in $O(s)$ space.

For $\Sigma =\{c_1,c_2,\ldots,c_\sigma\}$, define sequences of functions $\G=(\Delta_\rho[P^R_{c_1}],\Delta_\rho[P^R_{c_2}],\ldots,\Delta_\rho[P^R_{c_\sigma}])$
and $\F=(\Delta_\rho(T_{c_1}),\Delta_\rho(T_{c_2}),\ldots,\Delta_\rho(T_{c_\sigma}))$.
The algorithm initializes an instance of the procedure of \cref{lem:online-batch-algorithm} with $\G$ during the arrival of the first block.
During the  arrival of the $b$th block, the algorithm creates the batch $\F|_{\Phi_b}$ as follows: after the arrival of $T[i]$, if $T[i]\ne T[i-\rho]$, then the algorithm sets $\Delta_\rho[T_{T[i]}](i)$ to be $1$ and $\Delta_\rho[T_{T[i-\rho]}](i)$ to be $-1$.
During the arrival of the $(b+1)$th block, the batch $\F|_{\Phi_b}$ is processed using \cref{lem:online-batch-algorithm} in order to compute $[\F \otimes \G]|_{\Phi_b}$.
Finally, for  all $i\in \Phi_b$, the algorithm uses the buffer of $[T\otimes P]$ together with the values of $[\F\otimes G]|_{\Phi_b}$  in order to report  the values $[T\otimes P](i)=[\F \otimes \G](i)-[T\otimes P](i-2\rho)+2[T\otimes P](i-\rho)$ (see \cref{lem:crosscorrelation_computation}).

\paragraph{Complexities.}
For each incoming character of $T$, updating the text buffer and the functions $\Delta_\rho[T_c]$ for all $c\in \Sigma$ costs $\Oh(1)$ time since changes are needed in $\Delta_\rho[T_c]$ for at most two characters $c$. 
Since $\rho \le k$ is a $d$-period of both $P$ and $T$, \cref{obs:period-length-small-norm} yields $\|\G\|\le2(k+d)=O(k)$ and $\|\F\|\le 2(k+d)=O(k)$.
Thus, the initialization of the procedure of \cref{lem:online-batch-algorithm} costs $O(\|\G\|)=O(k)$ time,
and executing the procedure of \cref{lem:online-batch-algorithm} on the $b$th batch $\F|_{\Phi_b}$ costs $\Ohtilde((\|\F\|+\|\G\|)\sqrt{s})=\Ohtilde(k\sqrt{s})$ time.
So, applying the algorithm of \cref{lem:online-batch-algorithm} during the arrival of \emph{any} block 
costs $\Ohtilde(k+k\sqrt{s})=\Ohtilde(k\sqrt{s})$ time, which, by a standard de-amortization, is $\Ohtilde(\frac{1}{s}\cdot k\sqrt{s})=\tilde O(\frac{k}{\sqrt{s}})=\tilde O(\sqrt k)$ time per character.

Note that $\supp(\G)\subseteq[0\dd m+\rho]=[0\dd O(m)]$ and $\supp(\F)\subseteq[0\dd \frac 32m+\rho]=[0\dd O(m)]$.
Moreover, both $\F$ and $\G$ are sequences of $\sigma$ functions.
Hence, the total time cost of applying the algorithm of \cref{lem:online-batch-algorithm},  updating the buffers, and computing the functions $\Delta_\rho[T_c]$ is $\tilde O\left(m+\min \left(\|\F\|\cdot\|\G\|,\frac {m(\|\F\|+\|\G\|)}{\sqrt s},\frac {\sigma m^2}s\right)\right)= \tilde O\left(m+\min \left(k^2,\frac {mk}{\sqrt s},\frac {\sigma m^2}s\right)\right)$.

\paragraph{Delay.} For any index $i\in \Phi_b$, after the arrival of $T[i]$, at most $s$ character arrivals take place until the call to the procedure of \cref{lem:online-batch-algorithm} involving $\F|_{\Phi_b}$. Moreover, due to the de-amortization, the computation of all the results $[T\otimes P]|_{\phi_b}$ takes place during the arrivals of another $s$ characters.
Hence, the delay of the algorithm is at most $2s$ character arrivals.
\end{proof}

\section{Periodic Pattern and Text -- without Delay}\label{sec:periodic-everything-without-delay}

In this section, we show how to compute the distances between $P$ and substrings of $T$ without any delay, assuming that $\rho\le k$ is a $d$-period of both $P$ and $T$ for some $d=\Oh(k)$.
Our approach is to use the tail partition technique, described in \cref{sec:preliminaries}. To do so, we set $|P_{tail}|= 2s$ and use the algorithm of \cref{lem:kmm-periodic-everything-head} on $P_{head}$ (notice that $\rho$ is a $d$-period of both $P_{head}$ and $P_{tail}$).
The remaining task is to describe how to compute  $\HAM(P_{tail},T[i-2s+1\dd  i])$.

In the online version of the convolution summation problem, the algorithm is  given two sequences of functions $\F$ and $\G$, where $\supp(\F)\subseteq[0\dd n-1]$ and $\supp(\G)\subseteq[0\dd m-1]$.
The sparse representation of $\G$ is available for preprocessing, whereas $\F$ is revealed online index by index: At the $i$th update, the algorithm receives $\F|_{\{i\}}$, i.e., all the non-zero entries $f(i)$ for $f\in \F$,
and the task is to compute $[\F\otimes \G](i)$.
The procedure we use for this problem is based on the algorithm of Clifford et al.~\cite{CEPP11}.
Nevertheless, since we state the procedure in terms of the convolution summation problem, whereas~\cite{CEPP11} states the algorithm in terms of pattern matching problems, we describe the details in the proof.

\begin{restatable}[{Based on~\cite[{Theorem 1}]{CEPP11}}]{lemma}{lemblackbox}\label{lem:black-box}
	Let	$\mathcal F$ and $\mathcal G$ be two sequences of $t$ functions	each such that $\supp(\F)\subseteq[0\dd n-1]$, $\supp(\G)\subseteq[0\dd m-1]$, and $\delta = \max(\max_i \|\F|_{\{i\}}\|,\max_i\|\G|_{\{i\}}\|)$ is the maximum number of non-zero entries at a single index.
	There exists an online algorithm that upon receiving $\F|_{\{i\}}$ for subsequent indices $i$ computes $[\F\otimes\G](i)$ using $O(\delta m)$ space in $\Ohtilde(\sqrt{\delta(\|\F\|+\|\G\|)})$ time per index.
	Moreover, the total running time of the algorithm is $\tilde O(\min(n\delta+(\|\F\|+\|\G\|)\sqrt{n}, nt))$.
\end{restatable}
\begin{proof}
	Let $M = \ceil{\log (m+2)}-1$.
	For every $a\in[1\dd M]$, define $\Gamma_a=[2^a-2\dd 2^{a+1}-3]$, and for every $b\in\big[0\dd \floor{\tfrac n{2^{a-1}}}\big]$, define $\Phi_{a,b}=[(b-4)\cdot2^{a-1}+2\dd (b-1)\cdot 2^{a-1}+1]$.
	Notice that the intervals $\Gamma_a$ and $\Phi_{a,b}$ are of length $2^a$ and $3\cdot 2^{a-1}$, respectively.
	Moreover, each integer from $[0\dd m-1]$ appears in exactly one range $\Gamma_a$,
	and for every $a\in[1\dd M]$, each integer from $[0\dd n-1]$ appears in at most three ranges $\Phi_{a,b}$.
	
	Let $\G_a=\G|_{\Gamma_a}$ and $\F_{a,b}=\F|_{\Phi_{a,b}}$.
	The algorithm computes, for every  $a\in[1\dd M]$ and every $b\in\big[0\dd \floor{\frac n{2^{a-1}}}\big]$, the convolution $\G_a\otimes \F_{a,b}$, using the algorithm of \cref{lem:offline-convolution}.
	The computation starts right after the arrival of the information on the last index in $\Phi_{a,b}$, which is $(b-1)\cdot 2^{a-1}+1$, takes place during $2^{a-1}$ positions arrivals, and ends before the arrival of position $b\cdot 2^{a-1}+1$, i.e., before the algorithm ends processing position $b\cdot 2^{a-1}$.
	
	After the arrival of $\F|_{\{i\}}$, it is required to output $[\F\otimes\G](i)$.
	By linearity of convolutions and since intervals $\Gamma_a$ induce a partition of $[0\dd m-1]$, we have $[\F\otimes\G](i) = \sum_{a=1}^{M}[\F\otimes\G_a](i)$.
	Observe that  $\{i\}-\Gamma_a\subseteq \Phi_{a,\lfloor{\frac {i}{2^{a-1}}\rfloor}}$
	due to $(\lfloor\frac {i}{2^{a-1}}\rfloor-1)\cdot 2^{a-1}+1=(\lceil\frac {i+1}{2^{a-1}}\rceil-2)\cdot 2^{a-1}+1\ge (\frac {i+1}{2^{a-1}}-2)\cdot 2^{a-1}+1=i-(2^a-2)$
	and $(\lfloor\frac {i}{2^{a-1}}\rfloor-4)\cdot 2^{a-1}+2\le (\frac {i}{2^{a-1}}-4)\cdot 2^{a-1}+2 < i-(2^{a+1}-3)$. 
	Consequently, $[\F\otimes\G](i)=\sum_{a=1}^{M}[\F_{a,\floor{\frac i{2^{a-1}}}}\otimes\G_a](i)$.
	This summation can be performed while the algorithm processes position $i$, because the computation of $\F_{a,\lfloor\frac i{2^{a-1}}\rfloor}\otimes\G_a$ ends before the algorithm ends processing position $\lfloor\frac i{2^{a-1}}\rfloor\cdot 2^{a-1}\le i$.
	
	\paragraph{Complexities.}
	Recall that $\Psi(\mathcal {X,Y})$ is the time required for the computation of all the non-zero entries of  $\mathcal X\otimes \mathcal Y$ in the offline setting.
	The algorithm computes, for each $a\in[1\dd M]$ and $b\in[0\dd \lfloor\frac {m}{2^{a-1}}\rfloor]$, the convolution $\F_{a,b}\otimes \G_b$ in $\Psi(\F_{a,b}, \G_a)$ time, and this work is spread out to $2^{a-1}$ character arrivals, resulting in $\tilde O\left(\frac { \Psi(\F_{b,a},\G_a)}{2^a}\right)$ time per character involved.
	At any moment the algorithm executes for each $a\in[1\dd M]$ at most one process that computes the convolution of $\G_a$ with some $\F_{b,a}$. Consequently, the time cost per index of this computation is $\tilde O\left(\sum_{a=1}^{M} \frac {\max_b \{\Psi(\F_{a,b},\G_b)\}}{2^a}\right)$.
	The total running time is $\Ohtilde(\sum_{a=1}^M \sum_{b=0}^{\lfloor n / 2^{a-1}\rfloor }\Psi(\F_{a,b},\G_b))$.
	
	Below, we provide upper bounds for these expressions using \cref{lem:offline-convolution}.
	Recall that the functions in $\F_{a,b}$ are $f_j|_{\Phi_{a,b}}$ for $1\le j \le t$
	and the functions in $\G_{a}$ are $g_j|_{\Gamma_a}$ for $1\le j \le t$.
	Moreover,  $\diam(\F_{a,b})=\Oh(2^a)$, $\diam(\G_a)=\Oh(2^a)$,
	and $\diam(\F_{a,b}\otimes \G_{a})=\Oh(2^a)$.
	The definition of $\delta$ further yields
	$\|\F_{a,b}\|=\Oh(\delta\cdot 2^a)$ and $\|\G_{a}\|=\Oh(\delta\cdot 2^a)$.
	
	\paragraph{Time per index $\Ohtilde(\sqrt{\delta(\|\F\|+\|\G\|)})$.}
	We have
	$\Psi(\F_{a,b},\G_a)=\Ohtilde((\|\F_{a,b}\|+\|\G_a\|)\sqrt{2^a})=
	\Ohtilde(\min(\|\F\|+\|\G\|,\delta\cdot 2^a)\sqrt{2^a})
	=\Ohtilde(\sqrt{(\|\F\|+\|\G\|)\cdot \delta \cdot 2^a}\sqrt{2^a})
	=\Ohtilde(2^a\sqrt{(\|\F\|+\|\G\|) \delta})
	$ and consequently
	\begin{multline*}\tilde O\left(\sum_{a=1}^{M} \tfrac {\max_b \{\Psi(\F_{a,b},\G_b)\}}{2^a}\right)
	=\Ohtilde\left(\sum_{a=1}^{M} \tfrac {\max_b 2^a\sqrt{(\|\F\|+\|\G\|) \delta}}{2^a}\right)
	=\Ohtilde\left(\sum_{a=1}^{M} \sqrt{(\|\F\|+\|\G\|) \delta}\right)\\
	=\Ohtilde(M\sqrt{(\|\F\|+\|\G\|) \delta})=\Ohtilde(\sqrt{(\|\F\|+\|\G\|) \delta}).\end{multline*}

	\paragraph{Total time $\Ohtilde(n\delta+(\|\F\|+\|\G\|)\sqrt{n})$.}
	We have \[\Psi(\F_{a,b},\G_a)=\Ohtilde\left(\delta\cdot 2^a+\sum_{j=1}^t \min\big(\big|f_j|_{\Phi_{a,b}}\big|\big|g_j|_{\Gamma_a}\big|,2^a\big)\right)\] and consequently
	\begin{align*}
	\sum_{a=1}^M & \sum_{b=0}^{\lfloor n / 2^{a-1}\rfloor }\Psi(\F_{a,b},\G_a) \\ &=
	\Ohtilde\left(\sum_{a=1}^M \sum_{b=0}^{\lfloor n / 2^{a-1}\rfloor }\left(\delta\cdot 2^a +\sum_{j=1}^t \min\big(\big|f_j|_{\Phi_{a,b}}\big|\big|g_j|_{\Gamma_a}\big|,2^a\big)\right)\right) \\
	&= \Ohtilde\left(\sum_{a=1}^M \sum_{b=0}^{\lfloor n / 2^{a-1}\rfloor }\delta\cdot 2^a+\sum_{a=1}^M \sum_{j=1}^t \sum_{b=0}^{\lfloor n / 2^{a-1}\rfloor} \min\big(\big|f_j|_{\Phi_{a,b}}\big|\big|g_j|_{\Gamma_a}\big|,2^a\big)\right) \\
	&= \Ohtilde\left(\sum_{a=1}^M \tfrac{n}{2^a}\cdot \delta \cdot 2^a+\sum_{a=1}^M \sum_{j=1}^t \min\left(\sum_{b=0}^{\lfloor n / 2^{a-1}\rfloor} \big|f_j|_{\Phi_{a,b}}\big|\big|g_j|_{\Gamma_a}\big|,\sum_{b=0}^{\lfloor n / 2^{a-1}\rfloor} 2^a\right)\right) \\
	&= \Ohtilde\left(\sum_{a=1}^M n\delta +\sum_{a=1}^M \sum_{j=1}^t \min\left(\big|f_j\big|\big|g_j|_{\Gamma_a}\big|,\tfrac{n}{2^a}\cdot 2^a\right)\right) \\
	&= \Ohtilde\left(Mn\delta+M \sum_{j=1}^t \min\left(\big|f_j\big|\big|g_j\big|,n\right)\right)\\
	&= \Ohtilde\left(n\delta +\sum_{j=1}^t \sqrt{\big|f_j\big|\big|g_j\big|n}\right)
	= \Ohtilde\left(n\delta + \sum_{j=1}^t \left(\big|f_j\big|+\big|g_j\big|\right)\sqrt{n}\right)\\
	&= \Ohtilde\left(n\delta + \left(\sum_{j=1}^t \big|f_j\big|+\sum_{j=1}^t \big|g_j\big|\right)\sqrt{n}\right)
	= \Ohtilde\left(n\delta + \left(\|\F\|+\|\G\|\right)\sqrt{n}\right).
	\end{align*}
	
	\paragraph{Total time $\Ohtilde(n t)$.}
	We have
	$\Psi(\F_{a,b},\G_a)=\Ohtilde(t \cdot 2^a)$ and consequently
	\begin{multline*}
	\sum_{a=1}^M \sum_{b=0}^{\lfloor n / 2^{a-1}\rfloor }\Psi(\F_{a,b},\G_a) =
	\Ohtilde\left(\sum_{a=1}^M \sum_{b=0}^{\lfloor n / 2^{a-1}\rfloor }t \cdot 2^a\right) =
	\Ohtilde\left(\sum_{a=1}^M \tfrac{n}{2^a}\cdot t \cdot 2^a\right) =
	\Ohtilde\left(\sum_{a=1}^M nt\right)\\ =
	\Ohtilde\left(M nt\right) =
	\Ohtilde\left(nt\right).\qedhere
	\end{multline*}
\end{proof}

Using the algorithm of \cref{lem:black-box} and the reduction from computing Hamming distance to the convolution summation problem defined in \cref{sec:hd-to-csp}, we achieve the following lemma.

\begin{lemma}\label{lem:kmm-periodic-everything-tail}
	Suppose that there exists $\rho\le k$ which is a $d$-period of both $P$ and $T$ for some $d=O(k)$.
	Then, there exists a deterministic online algorithm for the $k$-mismatch problem  that uses $\tilde O(m)$ space and costs $\tilde O\left(n+\min\left(k\sqrt{n}, n\sigma\right)\right)$ total time. Moreover, the worst-case time cost per character is $\tilde O(\sqrt k)$.
\end{lemma}
\begin{proof}
The algorithm maintains a buffer of the last $2\rho$ values of $[T\otimes P]$
and a buffer of the last $\rho$ characters of $T$.
Define sequences of functions $\G=(\Delta_\rho[P^R_{c_1}],\Delta_\rho[P^R_{c_2}],\ldots,\Delta_\rho[P^R_{c_\sigma}])$ and $\F=(\Delta_\rho(T_{c_1}),\Delta_\rho(T_{c_2}),\ldots,\Delta_\rho(T_{c_\sigma}))$,
where $\Sigma =\{c_1,c_2,\ldots,c_\sigma\}$.
The algorithm initializes an instance of the procedure of  \cref{lem:black-box}.
After the arrival of $T[i]$, if $T[i]\ne T[i-\rho]$, then the algorithm sets $\Delta_\rho[T_{T[i]}](i)$ to be $1$ and $\Delta_\rho[T_{T[i-\rho]}](i)$ to be $-1$.
The algorithm transfers these values to the procedure of \cref{lem:black-box}, which responds with the value of $[\F\otimes\G](i)$.
Then, the algorithm reports $[T\otimes P](i)$, which equals $[\F\otimes\G](i)-[T\otimes P](i-2\rho)+2[T\otimes P](i-\rho)$ by \cref{lem:crosscorrelation_computation}. In this expression, the first term is returned by the procedure of \cref{lem:black-box} and the other two terms are retrieved from the buffer.

The update of the text buffer and $\Delta_\rho[T_c]$ for at most two characters $c$ after the arrival of any text character takes $O(1)$ time per character and $O(n)$ time in total.
Note that by \cref{obs:period-length-small-norm} since $\rho\le k$ is a $d$-period of both $P$ and $T$, we have that $\|\G\|\le2(d+k)=O(k)$ and $\|\F\|\le 2 (d+k)=O(k)$.
Moreover, $\delta = \max(\max_i \|\F|_{\{i\}}\|, \max_i \|\G|_{\{i\}}\|)\le 2$.
Consequently, the algorithm of \cref{lem:black-box} uses $\Ohtilde(m)$ space and costs $\Ohtilde(\sqrt{k})$ time per character and
$\Ohtilde(\min(n+k\sqrt{n},n\sigma))$ time total. All the other parts of the algorithm cost $O(1)$ time per character and $\Oh(n)$ time in total.
\end{proof}

Thus, by combining \cref{lem:kmm-periodic-everything-head} and \cref{lem:kmm-periodic-everything-tail} we achieve an algorithm the computes the Hamming distances up to $k$ for every substring of the text without delay, in the case where $\rho$ is a $d$-period of both $P$ and $T$.

\begin{lemma}\label{lem:kmm-periodic-everything}
	Suppose that there exists $\rho \le k$ which is a $d$-period of both $P$ and $T$ for some $d=O(k)$.
	Then, there exists a deterministic streaming algorithm for the $k$-mismatch problem with $n= \frac32m$ that, given an integer parameter $k\le s \le m$, uses $\tilde O(s)$ space, and costs $\tilde O\left(m+\min\left(k^2,\frac {mk}{\sqrt s},\frac {\sigma m^2}{s}\right)\right)$ total time and $\tilde O(\sqrt k)$ time per character in the worst case.
\end{lemma}
\begin{proof}
	Let $P_{tail}$ be the suffix of $P$ of length $2s$, and let $P_{head}$ be the complementary prefix of $P$.
	Note that $\rho$ is a $d$-period of both $P_{head}$ and $P_{tail}$.
	We run the algorithm of \cref{lem:kmm-periodic-everything-head} with $P_{head}$ as a pattern.
	The results of the algorithm of \cref{lem:kmm-periodic-everything-head} are added into a buffer of length $2s$.
	Thus, right after the arrival of $T[i]$, it is guaranteed that $\HAM(P_{head},T[i-|P|+1\dd i-|P_{tail}|])$ is already computed and available for the algorithm.	
	In addition, the algorithm of \cref{lem:kmm-periodic-everything-tail} is executed with $P_{tail}$ as pattern, which computes $\HAM(P_{tail},T[i-|P_{tail}|+1\dd i])$ right after the arrival of $T[i]$, if it is at most $k$.
	After the arrival of $T[i]$, the algorithm reports $\HAM(P,T[i-|P|+1\dd i])=\HAM(P_{head},T[i-|P|+1\dd i-|P_{tail}|])+\HAM(P_{tail},T[i-|P_{tail}|+1\dd i])$.	
	
The time per character of both the algorithm of \cref{lem:kmm-periodic-everything-head} and the algorithm of \cref{lem:kmm-periodic-everything-tail} is $\tilde O(\sqrt k)$.
The total time cost of the algorithm of \cref{lem:kmm-periodic-everything-head} is $\tilde O\left(m+\min\left(k^2,\frac {mk}{\sqrt s},\frac {\sigma m^2}{s}\right)\right)$.
Furthermore, the total time cost of the algorithm of \cref{lem:kmm-periodic-everything-tail} is $\tilde O\left(\min\left(m+k\sqrt{m}, m\sigma\right)\right)=\tilde O\left(m+\min\left(k^2,\tfrac{km}{\sqrt s},\tfrac{\sigma m^2}{s}\right)\right)$ due to $k\le s\le m$ and $m+k\sqrt{m}=\Oh(m+k^2)$.
Hence, the lemma follows.
\end{proof}

\section{Periodic Pattern and Arbitrary Text -- without Delay}\label{sec:periodic-pattern-aribtrary-text}

In this section, we generalize \cref{lem:kmm-periodic-everything} to the case where $\rho$ is not necessarily a $d$-period of $T$, but $\rho$ is still a $d$-period of $P$ and $n=\frac32m$.
We use ideas building upon Clifford et al.~\cite[Lemma 6.2]{CFPSS16} to show that there exists a substring of $T$, denoted $T^*$, such that $\rho$ is a $(2d+4k+\rho)$-period of $T^*$, and $T^*$ contains all of the $k$-mismatch occurrences of $P$ in $T$.
Our construction, specified below, exploits the following property of approximate periods.
\begin{observation}\label{obs:occ-periodicity}
	Let $X$ and $Y$ be two equal length strings. If $\rho$ is a $d$-period of $X$ and $\HAM(X,Y)\le x$ then $\rho$ is a $(d+2x)$-period of $Y$.
\end{observation}

Let $T_L$ be the longest suffix of $T[0\dd \frac12m-1]$ such that $\rho$ is a $(d+2k)$-period of $T_L$,
and let $T_R$ be the longest prefix of $T[\frac12m\dd n-1]$ such that $\rho$ is a $(d+2k)$-period of $T_R$. Finally, let $T^*$ be the concatenation $T^*=T_L\cdot T_R$.

\begin{lemma}\label{lem:Tstar}
	All the $k$-mismatch occurrences of $P$ in $T$ are contained within $T^*$.
	Moreover, $\rho$ is a $(2d+4k+\rho)$-period of $T^*$.
\end{lemma}
\begin{proof}
	The second claim follows directly from the fact that $T^*=T_L\cdot T_R$ is a concatenation
	of two strings with $(d+2k)$-period $\rho$ (the extra $\rho$ mismatches might occur at the boundary between $T_L$ and $T_R$). Henceforth, we focus on the first claim.

	We assume that $P$ has at least one $k$-mismatch occurrence in $T$; otherwise, the claim holds trivially.
	Let $T[\ell\dd r]$ be the smallest fragment of $T$ containing all the $k$-mismatch occurrences of $P$ in $T$ (so that the leftmost and the rightmost occurrences starts at positions $\ell$ and $r-m+1$, respectively).
	Our goal is to prove that $T[\ell\dd r]$ is contained within $T^*$.
	
	By \cref{obs:occ-periodicity}, since $T[\ell\dd \ell+m-1]$ is a $k$-mismatch occurrence of $P$ and $\rho$ is a $d$-period of $P$, it must be that $\rho$ is a $(d+2k)$-period of $T[\ell\dd \ell+m-1]$.
	In particular, since $\ell+m \ge \frac12m$, we have that $\rho$ is a $(d+2k)$-period of $T[\ell\dd \tfrac 12m-1]$.
	Hence, by its maximality, $T_L$ must start at position $\ell$ or to the left of $\ell$.
	Similarly, by \cref{obs:occ-periodicity}, since $T[r-m+1\dd r]$ is a $k$-mismatch occurrence of $P$ and $\rho$ is an $d$-period of $P$, it must be that $\rho$ is a $(d+2k)$-period of $T[r-m+1\dd r]$.
	In particular, since $r-m+1 \le n-m \le \frac12m$, we have that $\rho$ is a $(d+2k)$-period of $T[\frac12m\dd r]$.
	Hence, by its maximality, $T_R$ must end at position $r$ or to the right of $r$.
\end{proof}

The algorithm works in two high-level phases.
In the first phase, the algorithm receives $T[0\dd \tfrac 12m-1]$, and the goal is to compute $T_L$.
In the second phase, the algorithm receives $T[\frac 12m\dd n-1]$ and transfers $T^*=T_{L}\cdot T_R$ to the subroutine of \cref{lem:kmm-periodic-everything-head}.
The transfer starts with a delay of $|T_L|$ characters and a standard de-amortization speedup is applied to reduce the delay to $0$ by the time $2|T_L|\le m$ characters are transferred, which is before the subroutine of \cref{lem:kmm-periodic-everything-head} may start producing output.
The algorithm terminates as soon as it reaches the end of $T_R$, i.e., when it encounters
more than $d+2k$ mismatches in $T[\frac 12m\dd n-1]$.

\renewcommand{\L}{\mathsf{L}}
The following \emph{periodic representation} (similar to one by Clifford et al.~\cite{CKP19}) is used for storing substrings of $T$.

\begin{fact}\label{fct:perrep}
For every positive integer $\rho$, there exists an  algorithm that maintains a representation of  $T[\ell\dd r]$ and supports the following operations in $\Oh(1)$ time each:
\begin{enumerate}
\item\label{it:shrink} Change the representation to represent $T[\ell+1\dd  r]$ and return $T[\ell]$.
\item\label{it:grow} Given $T[r+1]$ and $T[r+1-\rho]$, change the representation  to represent $T[\ell\dd r+1]$.
\item\label{it:crop} Given $\ell'\ge \ell$ such that $T[i]=T[i+\rho]$ for $\ell \le i < \ell'$,  change the representation  to represent $T[\ell'\dd  r]$.
\end{enumerate}
If $\rho$ is a $d$-period of $T[\ell\dd r]$ then the space usage is $\Oh(d+\rho)$.
\end{fact}

\begin{proof}
$T[\ell \dd  r]$ is represented by a string $S$ of length $\rho$ such that $T[i]=S[i\bmod \rho]$ for $\ell \le i < \ell +\rho$,
and a list $\L=\{(i,T[i]): \ell+\rho \le i \le r, T[i-\rho]\ne T[i]\}$.
Notice that if $\rho$ is a $d$-period of $T[\ell\dd r]$, then this representation uses $O(d+\rho)$ space.

To implement operation~\eqref{it:shrink}, the algorithm first retrieves $T[\ell]=S[\ell\bmod \rho]$.
The algorithm then checks if the leading element of $\L$ is $(\ell+\rho,T[\ell+\rho])$.
If so, the algorithm removes this pair from $\L$ and sets $S[\ell \bmod \rho]=T[\ell+\rho]$.
To implement operation~\eqref{it:grow}, the algorithm compares $T[r+1]$ with $T[r+1-\rho]$.
If these values are different, then $(r+1, T[r+1])$ is appended to $\L$.
The implementation of operation~\eqref{it:crop} is trivial.
\end{proof}

\begin{lemma}\label{lem:kmm-periodic-pattern}
	Suppose that there exists $\rho\le k$ which is a $d$-period of $P$ for some $d=O(k)$.
	Then, there exists a deterministic streaming algorithm for the $k$-mismatch problem with $n= \frac32m$ that, given an integer parameter $k\le s \le m$, uses $\tilde O(s)$ space, and costs $\tilde O\left(m+\min\left(k^2,\frac {mk}{\sqrt s},\frac {\sigma m^2}{s}\right)\right)$ total time and $\tilde O(\sqrt k)$ time per character in the worst case.
\end{lemma}
\begin{proof}
	Based on the pattern $P$ and the period $\rho$, the algorithm initializes an instance $\ALG$ of the algorithm of \cref{lem:kmm-periodic-everything-head}. Then, the algorithm processes $T$ in two phases, while maintaining a buffer of $\rho$ text characters and a periodic representation of a suffix $T'$ of the already processed prefix of $T$.
	
	\paragraph{First Phase.}
	During the first phase, when the algorithm receives $T[0\dd \frac 12m-1]$, the suffix $T'$ is defined as the longest suffix for which $\rho$ is a $(d+2k)$-period.
	Suppose $T'$ is $T[\ell\dd  i-1]$ after processing $T[0\dd  i-1]$.
	The algorithm first appends $T[i]$ to $T'$, extending it to $T[\ell\dd  i]$
	using \cref{fct:perrep}\eqref{it:grow}. If $\rho$ is still a $(d+2k)$-period of $T'$,
	i.e., the list $\L$ has at most $d+2k$ elements, then the algorithm proceeds to the next character.
	Otherwise, $T'$ is first trimmed to $T[\ell'\dd  i]$, where $(\ell'+\rho,T[\ell'+\rho])$ is the first element of $\L$, using \cref{fct:perrep}\eqref{it:crop}, and then to $T[\ell'+1\dd  i]$ using \cref{fct:perrep}\eqref{it:shrink}. The latter operation decrements the size of $\L$ to $d+2k$.
	
	At the end of the first phase, $T'$ is by definition equal to $T_L$.
	The running time of the algorithm in the first phase is $O(1)$ per character,
	and the space complexity is $\Oh(d+k)=\Oh(k)$.
	
	\paragraph{Second Phase.}
	At the second phase, the algorithm receives $T[\frac m2\dd n-1]$,
	while counting the number of mismatches with respect to $\rho$.
	As soon as this number exceeds $d+2k$, which happens immediately after receiving the entire string $T_R$,
	the algorithm stops.
	As long as $T'$ is non-empty, each input character is appended to $T'$ using \cref{fct:perrep}\eqref{it:grow}, and the two leading characters of $T'$ are popped using \cref{fct:perrep}\eqref{it:shrink} and transferred to $\ALG$.
	Once $T'$ becomes empty (which is after $\ALG$ receives $2|T_L|\le m$ characters),
	the input characters are transferred directly to $\ALG$.
	This process guarantees that the input to $\ALG$ is $T^*$ and that $\ALG$ is executed with no delay by the time the first $m$ characters of $T^*$ are passed. By \cref{lem:Tstar}, all $k$-mismatch occurrences of $P$ in $T$ are contained in $T^*$, so all of these occurrences are reported in a timely~manner.

	Since $\rho$ is a $(2d+4k+\rho)$-period of $T^*$ (by \cref{lem:Tstar}) and $T'$ is contained in $T$, then $\rho$ is also a $(2d+4k+\rho)$-period of $T'$ at all times. Consequently, the space complexity is $\Oh(k)$ on top of the space usage of $\ALG$, which is $\Ohtilde(s)$. Thus, in total, the algorithm uses $\tilde O(s)$ space.
	
	The per-character running time is dominated by the time cost of $\ALG$, which is $\tilde O(\sqrt k)$.
	The total running time of the algorithm is also dominated by the total running time of $\ALG$, which, by \cref{lem:kmm-periodic-everything-head}, is $\tilde O\left(m+\min\left(k^2,\frac {mk}{\sqrt s},\frac {\sigma m^2}{s}\right)\right)$.
\end{proof}

The following corollary is obtained from \cref{lem:kmm-periodic-pattern} by the standard trick of splitting the text into $O(\tfrac nm)$ substrings of length $\frac 32m$ with overlaps of length $m$.
\begin{corollary}\label{cor:kmm-periodic-pattern}
	Suppose that there exists $\rho \le k$ which is a $d$-period of $P$.
	Then, there exists a deterministic streaming algorithm for the $k$-mismatch problem that, given an integer parameter $k\le s \le m$, uses $\tilde O(s)$ space and costs $\tilde O\left(n+\min\left(\frac {nk^2}m,\frac{nk}{\sqrt s},\frac{\sigma nm}s\right)\right)$ total time.
	Moreover, the worst-case time cost per character is $\tilde O(\sqrt k)$.
\end{corollary}

\section{Aperiodic Pattern and Arbitrary Text}\label{sec:aperiodic-pattern}

We first show how to compute the Hamming distance between $P$ and $T$ with delay of $k$ characters, in the case where $P$ is aperiodic. The algorithm appears in~\cite{GKP18} with small modifications. 
Then, in \cref{sec:aperiodic-pattern-without-delay} we show how to remove the delay using the tail partition technique, described in \cref{sec:preliminaries}.

\subsection{Aperiodic Pattern and Arbitrary text -- with Delay}\label{sec:aperiodic-pattern-delay}

In this section we prove the following lemma, appears in~\cite{GKP18} with small modifications.

\begin{lemma}[{Based on~\cite[Theorem 5]{GKP18}}]\label{lem:kmm-np-head}
	Suppose that the smallest $4k$-period of the pattern $P$ is $\Omega(k)$.
	Then, there exists a randomized streaming algorithm for the $k$-mismatch problem that uses $\tilde O(k)$ space and costs $\tilde O(1)$ time per character.
	The algorithm has delay of $k$ characters and is correct with high probability.
\end{lemma}

In the \emph{multi-stream dictionary matching problem}, a set $D=\{P_1,P_2,\ldots, P_d\}$ of length-$m$ patterns is given for preprocessing. In addition, there are $\alpha$ text streams $T_1,T_2,\ldots, T_\alpha$, and the goal is to report every occurrence of a pattern from $D$ in any $T_i$, for $1\le i \le \alpha$, as soon as the occurrence arrives.
An algorithm for the multi-stream dictionary matching problem is allowed to set up a read-only block of \emph{shared memory} during a preprocessing phase, whose contents depend solely on $D$, and $\alpha$ blocks of \emph{stream memory}, one for each text stream, to be used privately for each text stream as the stream is being processed.
The following theorem of Golan et al.~\cite{GKP18} will be useful in this section.
\begin{theorem}[{\cite[{Theorem 1}]{GKP18}}]\label{thm:same_length}
	There exists an algorithm for the multi-stream dictionary matching problem that uses $O(d\log m)$ words of shared memory, $O(\log m \log d)$ words of stream memory, and $O(\log m)$ time per character. All these complexities are in the worst-case, and the algorithm is correct with high probability.
\end{theorem}

The algorithm of \cref{lem:kmm-np-head} has two conceptual phases.
The first phase serves as a filter by establishing for every $q$ whether $\HAM(P,T[q-m+1\dd q])>2k$ or not.
The second phase, which computes the exact Hamming distance, is only guaranteed to work when $\HAM(P,T[q-m+1\dd q])\le 2k$.
Notice, however, that the algorithm must run both phases concurrently since we are only able to establish whether $\HAM(P,T[q-m+1\dd q])>2k$ or not after the $q$th character has arrived, and if we were to begin the computation for the second phase only for locations that pass the filter, we would need access to a large portion of $T[q+m+1\dd q]$ at this time.

\paragraph{Offset texts and patterns.}
Let $p$ be a prime number, which for simplicity is assumed to divide $m$.
Consider the conceptual matrix $M^p=\{m^p_{x,y}\}$ of size ${\frac{m} {p}} \times p$ where $m^p_{x,y} = P[x\cdot p + y]$.
For any integer $0\le r< p$, the $r$th column corresponds to an \emph{offset pattern} $P_{p,r}=P[r]P[r+p]P[r+2p]\cdots P[m-p+r]$.
Notice that some offset patterns with the same prime $p$ but different values of $r$ might be equal.
Let $\Gamma_p = \{P_{p,r}\,\mid \,0\le r<p \}$ be the set of all the offset patterns.
Each unique offset pattern is associated with a unique id; the set of unique ids is denoted by $ID_p$.
The columns of $M^p$ define a \emph{column pattern} $P_p$ of length $p$, where the $i$th character is the unique id of the $i$th column.
We also partition $T$ into $p$ \emph{offset texts}, where for every $0\le r< p$ we define $T_{p,r}=T[r] T[r+p] T[r+2p]\cdots$.

Using the multi-stream dictionary matching of \cref{thm:same_length}, the algorithm finds occurrences of offset patterns from $\Gamma_p$ in each of the offset texts\footnote{Notice that if $p$ does not divide $m$, then we would use two instances of the multi-stream dictionary matching, since all of the patterns in $\Gamma_p$ have length either $\floor{m/p}$ or $\ceil{m/p}$.}.
When the character $T[q]$ arrives, the algorithm passes $T[q]$ to the stream of $T_{p,q\bmod p}$, which is one of the input streams used in the instance of the algorithm of \cref{thm:same_length}.
The algorithm also creates a single streaming \emph{column text} $T_p$ whose characters correspond to the ids of offset patterns as follows.
If one of the offset patterns is found when $T[q]$ is passed to $T_{p,q\modulo p}$, then its unique id is the $q$th character in $T_p$.
Otherwise, the algorithm uses a dummy character for the $q$th character in $T_p$.

Notice that there is an occurrence of $P$ in $T$ after the arrival of the $q$th text character if and only if there is an occurrence of $P_p$ in $T_p$ at that point in time.
Nevertheless, if $\HAM(P,\allowbreak T[q-m+1\dd q]) > 0$, then we cannot guarantee that $\HAM(P_p,T_p[q-p+1\dd q]) = \HAM(P,T[q-m+1\dd q])$ since there could be several mismatches that are mapped to the same offset text.
Following the terminology of Clifford et al.~\cite{CFPSS16}, a mismatch between pattern location $i$ and text location $j$ is said to be \emph{isolated relative to $q$ and $p$} if and only if it is the only mismatch between the pattern offset that contains $P[i]$ and the corresponding offset text that contains $T[j]$.
The idea behind the algorithm is to use several different values of $p$ in a way that guarantees that each mismatch is isolated for at least one choice of $p$.

\subsubsection{Filtering}

The goal of this phase it to establish for every $q$ whether $\HAM(P,T[q-m+1\dd q])>2k$.
The following lemma of Clifford et al.~\cite{CFPSS16} is useful for our algorithm.

\begin{lemma}[{\cite[Lemmas 5.1 and 5.2]{CFPSS16}}]\label{lem:filter_kmm}
	Let $\Pi$ be a uniformly random set of $\log m$ prime numbers between $k\log^2m$ and $34k\log^2 m$.
	If $\HAM(P,T[q-m+1\dd q])\le2k$, then $\HAM(P_p,T_p[q-p+1\dd q])\le \HAM(P,T[q-m+1\dd q])$ for every prime $p\in \Pi$.
	Moreover, if $\HAM(P,T[q-m+1\dd q])>2k$, then $\max_{p\in \Pi} \{\HAM(P_p,T_p[q-p+1\dd q])\} > \frac 5 4 k$ with probability at least $1-\frac 1{4m^2}$.
\end{lemma}

For every prime number $p \in \Pi$, the algorithm creates an instance of the multi-stream dictionary matching algorithm of \cref{thm:same_length} where the dictionary contains all the offset patterns $P_{p,r}$, and the text streams are the offset texts $T_{p,r}$.
When a new character $T[q]$ arrives, the character is passed on to $\log m$ streams (offset texts), one for each prime number.
For every prime $p\in \Pi$, the algorithm computes a $(1-\eps)$ approximation to $\HAM(P_p,T_p[q-p+1\dd q])$, for $\eps<1/5$,  using the black-box technique of~\cite{CEPP11} on the approximate Hamming distance algorithm of either Karloff~\cite{Karloff93} or Kopelowitz and Porat~\cite{KP15, KP18}.
This costs $\tilde{O}(1)$ time per character per prime in $\Pi$.
If $\HAM(P,T[q-m+1\dd q])>2k$, then by \cref{lem:filter_kmm}  the maximum over all primes in $\Pi$ is larger than $k$ with high probability.
For any prime number $p\in \Pi$, the corresponding dictionary contains $\tilde O(k)$ patterns and $\tilde O(k)$ streams, so the total space usage of each multi-stream dictionary matching instance is $\tilde O(k)$.
Summing over $O(\log m)$ prime numbers, the total space usage of this component is still $\tilde O(k)$, and the time per character is $\tilde{O}(1)$.

\subsubsection{Exact Computation}

We assume from now that for every location $q$ that passed the filtering phase, we have $\HAM(P,T[q-m+1\dd q])\le 2k$.
The goal is to compute the exact Hamming distance for such locations.
To do so, we distinguish for each $q$ which mismatches are isolated relative to some $p\in \Pi$ and $q$.
This task is accomplished in two steps: First, we establish for each offset pattern $P_{p,r}$ and corresponding offset text $T_{p,r'}$ at time $q$ that created a mismatch in $P_p$ versus $T_p$ whether the Hamming distance of $P_{p,r}$ and the suffix of $T_{p,r'}$ of length $|P_{p,r}|$  is $0$, $1$, or more.
Secondly, if the Hamming distance is exactly $1$, then we find the location of the only mismatch.

The first step is executed by applying a filter similar to the one used in the first phase, but this time we have $k=1$.
In particular, we make use of the following lemma.
\begin{lemma}\label{lem:primes_for_one_mismatch}
	There exists a set $Q$ of $\Theta(\frac {\log m}{\log\log m})$ prime numbers such that $\prod_{p\in Q}p>m$ and $p=\Theta(\log m)$ for each $p\in Q$
\end{lemma}

\begin{proof}
	If $m$ is smaller than some constant, to be defined shortly, then the claim is straightforward by choosing $Q$ to contain only this constant.
	We define $Q$ to be the set of prime numbers in the range $[\ln m, 5\ln m]$. Due to~\cite[Claim 1]{RSL62}, we have that for any $m\ge 17$, the number of prime numbers in the range $[1\dd m]$, denoted by $\pi(m)$, satisfies the inequality $\frac m{\ln m}<\pi(x)<1.256 \frac m{\ln m}$.
	Therefore, 	assuming $m>24154957>e^{17}$, we have that
	\[|Q|>\tfrac{5\ln m}{\ln(5\ln m)}-1.256\tfrac{\ln m}{\ln\ln m} = \tfrac{5\ln m}{\ln5+\ln\ln m}-1.256\tfrac{\ln m}{\ln\ln m}.\]
	Since $\ln\ln m>\ln 17 > \ln 5$, we have
	\[|Q|>\tfrac{5\ln m}{2\ln\ln m}-1.256\tfrac{\ln m}{\ln\ln m}>\tfrac{\ln m}{\ln\ln m}.\]
	Hence, $Q$ is a set of size $\Theta(\frac {\log m}{\log \log m})$ and each $p\in Q$ is of size $\Theta(\log m)$.
	Finally, \[\prod_{p\in Q}p>\prod_{p\in Q} \ln m = (\ln m)^{\frac{\ln m}{\ln\ln m}}=(e^{\ln\ln m})^{\frac{\ln m}{\ln\ln m}}=e^{\ln m}=m.\qedhere\]
\end{proof}

From the properties of $Q$, we are able to complete the first step with high probability; for more details, see~\cite[Lemma 4.3]{CFPSS16}.
The second step is executed by using the techniques of Porat and Porat~\cite[Section 6]{PP09}, which also applies the Chinese Remainder Theorem.

The following fact appeared in~\cite{CFPSS16}, and is useful for the analysis of the time complexity.
\begin{fact}[{Based on~\cite[{Fact 3.1}]{CFPSS16}}]\label{lem:head_diff}
	If $\rho$ is the smallest $d$-period of a pattern $P$, then the $d/2$-mismatch
	occurrences of $P$ in any text $T$ start at least $\rho$ positions apart.
\end{fact}

\paragraph{Complexities.}
The space usage is dominated by the space usage of all the instances of the multi-stream dictionary matching algorithm, which is $\tilde O(k)$.
The space usage of both components is linear in the number of offset patterns, which is $\tilde O(k)$.
The first component takes $\tilde O(1)$ time per character.
The second components takes $\tilde O(1)$ time per character, and additional $\tilde O(k)$ time per position where we compute the isolated mismatch locations ($\tilde O(1)$ time per location).


By \cref{lem:head_diff}, since the smallest $4k$-period of $P$ is $\rho=\Omega(k)$, the occurrences of the pattern with up to $2k$ mismatches must be $\Omega(k)$ locations apart.
Since the only locations in $T$ whose Hamming distance with $P$ is at most $2k$ are able to pass the filter phase, then the second phase is executed only once every $\Omega(k)$ characters.

Therefore, whenever the first component of the algorithm finds a $2k$ occurrence of $P$, the $\tilde O(k)$ work of the second component, which computes the exact Hamming distance, is de-amortized to the next $k$ character arrivals.
Thus, the running time per character is $\tilde O(1)$.
The delay introduced by the algorithm is at most $k$, completing the proof of \cref{lem:kmm-np-head}.\qed%

\subsection{Aperiodic Pattern and Arbitrary text -- without Delay}\label{sec:aperiodic-pattern-without-delay}
In order to improve the algorithm of \cref{lem:kmm-np-head} to report results without any delay, we use ideas similar to those introduced in \cref{sec:periodic-everything-without-delay}.

\begin{lemma}\label{lem:kmm-np-all}
	Suppose that the smallest $6k$-period of the pattern $P$ is $\Omega(k)$.
	Then, there exists a randomized streaming algorithm for the $k$-mismatch problem that uses $\tilde O(k)$ space and costs $\tilde O(1)$ time per character.
	The algorithm is correct with high probability.
\end{lemma}
\begin{proof}
	Let $P_{tail}$ be the suffix of $P$ of length $2k$ and let $P_{head}$ be the complementary prefix of $P$. Since the smallest $6k$-period of $P$ is $\Omega(k)$, $|P_{tail}|=2k$, and $6k-2k=4k$, the smallest $4k$-period of $P_{head}$ is also $\Omega(k)$.
	Thus, we execute the procedure of \cref{lem:kmm-np-head} with $P_{head}$.
	Then, whenever the procedure reports  $\HAM(P_{head},T[i-|P|+1\dd i-2k])$ to be at most $k$, the algorithm starts a process that computes $\HAM(P_{tail},T[i-2k+1\dd i])$.
	The procedure of \cref{lem:kmm-np-head} reports $\HAM(P_{head},T[i-|P|+1\dd i-2k])$ before $T[i-k+1]$ arrives.
	Hence, there are still at least $k$ character arrivals until $\HAM(P,T[i-|P|+1\dd i])$ has to be reported.
	During these character arrivals, the computation of $\HAM(P_{tail},T[i-2k+1\dd i])$ is done simply by comparing pairs of characters. The total time of this computation is $O(k)$, and by standard de-amortization, this is $O(1)$ time per character during the arrival of the $k$~characters.
	
	Since the smallest $4k$-period of $P$ is $\Omega(k)$, by \cref{lem:head_diff} we have that any two $k$-mismatch occurrences of $P$ in $T$ are at distance  $\Omega(k)$.
	Therefore, the maximum number of processes computing distances to $P_{tail}$ at any time is $O(1)$.
	Thus, the time cost per character of the algorithm is dominated by the procedure of \cref{lem:kmm-np-head}.
	Furthermore, the total time cost of computing the distances to $P_{tail}$ is $O(n)$ and the overall time cost of the algorithm is dominated by the procedure of \cref{lem:kmm-np-head}.
\end{proof}

\section{Proof of Main Theorem}
We conclude the paper with a proof of \cref{thm:kmm-total}, which is our main result. In the preprocessing, the shortest $6k$-period $\rho$ of the pattern $P$ is determined.
If $\rho\le k$, then the text is processed using \cref{cor:kmm-periodic-pattern}.
This procedure uses $\Ohtilde(s)$ space and costs $\Ohtilde(\sqrt{k})$ time per character and  $\tilde O\left(n+\min\left(\frac {nk^2}m,\frac{nk}{\sqrt s},\frac{\sigma nm}s\right)\right)$ time in total.
Otherwise, the text is processed based on the \cref{lem:kmm-np-all}.
The space complexity in this case is $\Ohtilde(k)=\Ohtilde(s)$, whereas the running time is $\Ohtilde(1)=\Ohtilde(\sqrt{k})$ per character and $\Ohtilde(n)$ in total. This completes the proof of \cref{thm:kmm-total}.\qed

\bibliography{Refs}

\appendix

\end{document}